\documentclass[a4paper,UKenglish,cleveref, autoref, thm-restate]{lipics-v2019}

\graphicspath{
	{./graphics/legends/}
	{./graphics/figure1/}
	{./graphics/figure2/}
	{./graphics/figure3/}
	{./graphics/figure6/}
	{./graphics/tables/}
}

\usepackage[procnumbered,ruled,vlined,linesnumbered]{algorithm2e}

\newcommand{\textin}{\operatorname{in}}
\newcommand{\textout}{\operatorname{out}}
\newcommand{\vol}{\operatorname{vol}}
\newcommand{\textdeg}{\operatorname{deg}}
\renewcommand{\Pr}{\operatorname{Pr}}
\newcommand{\Ex}{\mathbb{E}}
\usepackage{dsfont}
\newcommand{\indicator}{\mathds{1}} 
\newcommand{\cA}{\mathcal{A}}

\declaretheorem{property} 
\declaretheorem{observation} 

\title{Engineering Nearly Linear-Time Algorithms for Small Vertex
  Connectivity} 

\author{Max Franck}{Department of Computer Science, Aalto University,
  Espoo, Finland}{max.franck@aalto.fi}{https://orcid.org/0000-0003-3583-8033}{}
\author{Sorrachai Yingchareonthawornchai}{Department of Computer Science, Aalto University, Espoo, Finland}{sorrachai.yingchareonthawornchai@aalto.fi}{https://orcid.org/0000-0002-7169-0163}{}

\authorrunning{M. Franck and S. Yingchareonthawornchai} 

\Copyright{Max Franck and Sorrachai Yingchareonthawornchai} 



\category{} 

\relatedversion{} 

\supplement{The source code is available at \\ \url{https://github.com/untellect/local-vertex-connectivity}}

\acknowledgements{This project has received funding from the European Research Council (ERC) under the European Union’s Horizon 2020 research and innovation programme under grant agreement No  759557}

\nolinenumbers 

\hideLIPIcs  

\EventEditors{David Coudert and Emanuele Natale}
\EventNoEds{2}
\EventLongTitle{19th International Symposium on Experimental Algorithms (SEA 2021)}
\EventShortTitle{SEA 2021}
\EventAcronym{SEA}
\EventYear{2021}
\EventDate{June 7--9, 2021}
\EventLocation{Nice, France}
\EventLogo{}
\SeriesVolume{190}
\ArticleNo{22}

\begin{document}

\maketitle

\begin{abstract}
Vertex connectivity is a well-studied concept in graph theory with numerous applications. A graph is $k$-connected if it remains connected after removing any $k-1$ vertices. The vertex connectivity of a graph is the maximum $k$ such that the graph is $k$-connected. There is a long history of algorithmic development for efficiently computing vertex connectivity. Recently, two near linear-time algorithms for small $k$ were introduced by [Forster et al. SODA 2020]. Prior to that, the best known algorithm was one by [Henzinger et al. FOCS'96] with quadratic running time when $k$ is small.

In this paper, we study the practical performance of the algorithms by
Forster et al. In addition, we introduce a new heuristic on a key
subroutine called local cut detection, which we call degree counting.
We prove that the new heuristic improves space-efficiency (which can be good for caching purposes) 
and allows the subroutine to terminate earlier. 
According to experimental results on random graphs
with planted vertex cuts, random hyperbolic graphs, and real world
graphs with vertex connectivity between 4 and 15, the degree counting
heuristic offers a factor of 2-4 speedup over the original
non-degree counting version for most of our data. It also outperforms the previous
state-of-the-art algorithm by Henzinger et al. even on relatively small graphs.
\end{abstract}

\section{Introduction}

Given an undirected graph, the \textit{vertex connectivity problem} is to compute the minimum size of
a vertex set $S$ such that after removing $S$, the remaining graph is
disconnected or a singleton.  Such a vertex-set is called
a \textit{minimum vertex cut}.
 Vertex connectivity is well-studied concept in graph theory with
applications in many fields. For example, for network reliability \cite{GoldschmidtJL94,LiuCL00}, a
minimum vertex-cut has the highest chance to disconnect the network
assuming each node fails independently with the same probability; in sociology, vertex connectivity of a
social network measures social cohesion \cite{SocialCohesion}.

There is a long history of algorithmic development for efficiently computing
vertex connectivity (see \cite{Nanongkai} for more elaborated
discussion of algorithmic development).
Let $n$ and $m$ be the number of vertices and edges respectively in
the input graph. The time complexity for computing vertex connectivity has been
$O(n^2)$ since 1970 \cite{Kleitman1969methods} even for the special
case where the connectivity is a constant until very recently, when
\cite{Forster} introduced randomized (Monte Carlo)\footnote{With at most $\frac{1}{n^c}$ error rate for any constant $c$.}
algorithms to compute vertex connectivity in time $O(m+ n\kappa^3\log^2n)$
(for undirected graphs) where $\kappa$ is the
vertex connectivity of the graph. The algorithm follows the framework
by \cite{Nanongkai}.  This makes progress
toward the conjecture (when $\kappa$ is a constant) by Aho, Hopcroft and Ullman \cite{ahoHU74} (Problem
5.30) that there exists a linear time algorithm for
computing vertex connectivity.  Before that, the state-of-the-art
algorithm was due to \cite{HRG}, which runs in time $O(n^2\kappa\log n)$.

In this paper, we study the practical performance of the near-linear time
algorithms by \cite{Forster} for small vertex connectivity. We briefly describe their framework and point out the potential
improvement of the framework. \cite{Nanongkai}  provide a fast reduction from
vertex connectivity to a subroutine called \textit{local vertex-cut
  detection}.  Roughly speaking, the framework deals with two extreme
cases: detecting balanced cuts and unbalanced cuts.  The balanced cuts can be
detected using (multiple calls to) a standard $st$-max flow algorithm; the unbalanced cuts
can be detected using (multiple calls to) local vertex-cut
detection. Reference \cite{Forster} follow the
same framework and observe that local vertex-cut detection can be further reduced to
another subroutine called \textit{local edge-cut detection} as well as provide fast edge cut detection
algorithms that finally prove the near-linear time vertex connectivity 
algorithm for any constant $\kappa$. The full algorithm is discussed in \Cref{appendix:fullocalvc}. From our internal testing, we observe
that, overall the framework, the performance bottleneck is on the
local edge detection algorithm.

Therefore, our focus is on speeding up the local edge-cut detection
algorithm. To define the problem precisely, we  first set up notations. Let $G =
(V,E)$ be a directed graph. Let $E(S,T)$ be the set of edges from vertex-set $S$
to vertex-set $T$. For any vertex-set $S$, let
$\vol^{\textout}(S) := \sum_{v \in S}\textdeg^{\textout}(v)$ denote
the volume of $S$ which is total number of edges originating in $S$.
Undirected edges are treated as one directed edge in each direction.
We now define the interface of the local edge-cut
detection algorithm.

\begin{definition} \label{def:localec}
An algorithm $\mathcal{A}$ is LocalEC if it takes as input a vertex $x$ of a
graph $G = (V,E)$, and two parameters $\nu, k$ such that $\nu k = O(|E|)$, and output in the following manner:
\begin{itemize}
  \item either output a vertex-set $S$ such
  that $x \in S$ and $|E(S,V \setminus S)| < k$  or,
  \item the symbol $\bot$ certifying that there is no non-empty
    vertex-set $S$ such that
    \begin{align} \label{eq:local cut exists}
     x \in S, \vol^{\textout}(S) \leq \nu, \mbox{ and } |E(S, V \setminus S)|
      <k.
    \end{align}
  \end{itemize}
  The algorithm is allowed to have bounded
one-sided error in the following sense. If there is a non-empty vertex-set $S$
satisfying \Cref{eq:local cut exists}
then  $\bot$ is returned with probability at most $1/2$.
\end{definition}

Reference \cite{Forster} introduced two LocalEC algorithms with the running time
$O(\nu k^2)$. The algorithms are very simple: they use repeated DFS
(depth-first search) with different conditions for early termination. We note
that this running time is enough to get a near-linear time algorithm
for small connectivity  using the framework by \cite{Nanongkai}.

\textbf{Our Results and Contribution.}  We introduce a heuristic called \textit{degree
  counting} that is applicable to both variants of LocalEC in
\cite{Forster}, which we call Local1+ and Local2+. We prove that the degree counting heuristic version is more
space-efficient in terms of \textit{edge-query} complexity and
\textit{vertex-query} complexity. Edge-query complexity is defined as the number
of edges that the algorithm accesses, and vertex-query
complexity is defined as the number of vertices that the 
algorithm accesses. The results are shown in Table \ref{table:new
 results}. These complexity measures can be relevant in practice. For
example, an algorithm with low query complexity may be able to store
the accessed data in a smaller cache than an algorithm with high query complexity.

\begin{table}[!h]
\caption{Comparisons among various implementation of LocalEC
  algorithms. Local1+ denotes Local1 with the degree counting heuristic. Similarly, Local2+ denotes Local2 with the degree counting heuristic. }
\label{table:new results}
\centering
\begin{tabular}{|c|c|c|c|c|}
\hline
LocalEC Variants & Time         & Edge-query   & Vertex-query & \multicolumn{1}{l|}{Reference} \\ \hline
Local1           & $O(\nu k^2)$ & $O(\nu k^2)$ & $O(\nu k^2)$ & \cite{Forster}                         \\ \hline
Local1+          & $O(\nu k^2)$ & $O(\nu k^2)$ & $O(\nu k)$   & This paper                         \\ \hline
Local2           & $O(\nu k^2)$ & $O(\nu k)$   & $O(\nu k)$   & \cite{Forster}                     \\ \hline
Local2+          & $O(\nu k^2)$ & $O(\nu k)$   & $O(\nu)$     & This paper                     \\ \hline
\end{tabular}
\end{table}

We conducted experiments on three types of undirected graphs: (1) graphs with planted cuts where we have control over size and volume of the cuts, and (2) random hyperbolic graphs, and (3) real-world
networks. We denote LOCAL1, LOCAL1+, and LOCAL2+ to be the same
local-search based vertex connectivity algorithm \cite{Forster} (see
\Cref{appendix:fullocalvc} for details) except that the unbalanced part is implemented with different LocalEC
algorithms using Local1, Local1+, Local2+, respectively. We use Local1 as a baseline for LocalEC algorithms. We denote HRG to be the preflow-push-relabel-based algorithm by \cite{HRG}. We implement
HRG as a baseline because when $k$ is small (say
$k = O(1)$) HRG is the fastest known alternative to \cite{Forster,
  Nanongkai}. The implementation detail can be found in
\Cref{appendix:hrg}. By sparsification algorithm \cite{sparsification}, we can assume that the input graph size depends on $n$ and $k$. The
following summarize the key finding of our empirical studies.

\begin{enumerate}
\item \textbf{Internal Comparisons (\Cref{sec:effectiveness}).} We compare three LocalEC algorithms (Local1, Local1+, Local2+). According to the experiments (\Cref{edgesovervk}), for any $\nu$
parameter, Local1+ and Local2+ visit \textit{significantly} fewer edges than Local1. Also, Local2+ visits slightly fewer edges than Local1+ overall. The degree counting is also very effective
at low volume parameter. When plugging into full vertex connectivity algorithms, the degree counting heuristics (LOCAL1+ and LOCAL2+) improve the performance over non-degree counting counter part
(LOCAL1) by a factor 2 to 4 for most data used in our experiments, although for some larger graphs the speedup was noticeably larger. The greatest observed speedup over LOCAL1 is 18.4x for LOCAL2+ at
$n=100000$, $\kappa_G=16$. For graphs of this size, LOCAL2+ performs slightly better than LOCAL1+. Finally, according to CPU sampling, the local search is the main bottleneck for the performance
of LOCAL1 at roughly at least $90\%$ for large instances. On the other hand, for the degree counting versions (LOCAL1+ and LOCAL2+), the CPU usage of local search part is improved to be almost the same
as the other main component (i.e., finding a balanced cut using the Ford-Fulkerson's max-flow algorithm).
\item \textbf{Comparisons to HRG.} We compare four vertex connectivity algorithms, namely HRG, LOCAL1, LOCAL1+, LOCAL2+. For planted cuts (\Cref{sec:planted cuts}), LOCAL1, LOCAL1+, and LOCAL2+ scale
with $n$ much better than HRG when $\kappa_G$ is fixed. In particular, LOCAL1+ and LOCAL2+ start to outperform HRG on graphs as small as $n \leq 500$ (when $\kappa \leq 15$).
 For random hyperbolic graphs (\Cref{sec:random hyperbolic graphs}), HRG performs much better than on the planted cut instances, but is still
outperformed relatively early. In particular, LOCAL1+ and LOCAL2+ outperform HRG for $n \geq 5000$ when $\kappa \leq 12.$ In real-world graphs (\Cref{sec:real world}), LOCAL1+ and LOCAL2+ are the
fastest among the four algorithms with LOCAL2+ being slightly faster than LOCAL1+. We also observe that the performance of all four algorithms is very similar on part of the real world dataset and
graphs with planted cuts with the same size and vertex connectivity.
\end{enumerate}

\textbf{Organization.} We discuss related work in \Cref{sec:related}, and preliminaries in \Cref{sec:prelim}. Then, we review two variants of LocalEC algorithms (Local1,Local2) \cite{Forster}, and
describe new degree counting heuristic versions (Local1+, Local2+) in \Cref{sec:localec}. Then, all the experimental results are discussed in \Cref{sec:experiment}. We conclude and discuss future work
in \Cref{sec:conclusion}.

\section{Related Work} \label{sec:related}

\textbf{Fast Vertex Connectivity Algorithms.}  We consider a decision
version where the problem is to decide if $G$ has a vertex cut of
size at most $k-1$ (the general vertex connectivity can be solved
using a binary search on $k$).  We highlight only recent
state-of-the-art algorithms. For more elaborated discussion, see
\cite{Nanongkai}. When $k = O(1)$, the fastest known algorithm is by
\cite{Forster} with running time $O(m+ nk^3 \log^2n)$. The algorithm
is based on local search approach. For larger $k$, the fastest known
algorithm are based  on preflow-push-relabel by \cite{HRG} with the running time $O(n^2k
 \log n)$, and based on algebraic techniques by \cite{LinialLW88} with the running time $O(n^{\omega} \log ^2n + k^{\omega}n
\log n)$ where $\omega$ denotes the matrix multiplication exponent,
currently $\omega \leq 2.37286$ \cite{abs-2010-05846}.  When $k$ is small (say
$k = O(1)$), the preflow-push-relabel-based algorithm  by \cite{HRG}
is the fastest alternative to \cite{Forster, Nanongkai}. Therefore, we
implement the preflow-push-relabel-based algorithm \cite{HRG} as a baseline for performance comparisons.  We note
both all aforementioned algorithms are randomized. Deterministic
algorithms are much slower than the randomized ones. The fastest known
deterministic algorithms are by \cite{Gabow06} for large $k$ and by
\cite{abs-1910-07950} for $k= O(1)$. 

\textbf{Deciding $(k,s,t)$-Vertex Connectivity.}   We mention another related
problem which is to decide if the there is a vertex cut separating $s$ and
$t$ of size at most $k-1$. By a standard reduction \cite{Splitgraph}, it can be
solved by $st$-maximum flow.  $st$-maximum flow can be solved in time
$O(mk)$ by augmenting paths algorithm by Ford-Fulkerson algorithm \cite{FordFulkerson}. For larger
$k$, a simple blocking flow algorithm by \cite{Dinitz06} runs in time
$O(m\sqrt{n})$.  The current state-of-the art algorithms are
$O(m^{4/3+o(1)})$-time algorithm by \cite{abs-2003-08929}, and $\tilde
O(m+n^{1.5})$-time\footnote{$\tilde O(f(n)) = O(\text{poly}(\log
  n)f(n)).$} algorithm by \cite{BrandLNPSS0W20}. Note that when $k$ is
small (e.g., $k = O(1)$), then Ford-Fulkerson algorithm \cite{FordFulkerson} is the fastest, and we thus
implement Ford-Fulkerson algorithm  as a subroutine to find vertex cut
for the balanced case.

\textbf{Local Search.}  There are quite a few local search algorithm
with different running time.  The first LocalEC algorithm by \cite{Chechik} has
running time of $O(\nu k^k)$. \cite{Forster} introduced a new local
search algorithm with improved time $O(\nu k^2)$.  \cite{Forster} also
provide a reduction to local vertex cut detection problem, which we called LocalVC (similar to \Cref{def:localec}, but uses
vertex cut instead of edge cut).  Therefore, there is a LocalVC
algorithm with running time $O(\nu k^2)$. This improved the previous
bound for LocalVC with running time $O(\nu^{1.5} k)$ by
\cite{Nanongkai} when $k$ is small.   For our purpose, when $k$ is
small (say $k = O(1)$), the algorithm by \cite{Forster} is the
fastest, and thus we consider the LocalEC algorithm by
\cite{Forster}.

\textbf{Implementation and Experimental Studies.}   To the best of our
knowledge, this paper is the first experimental study on vertex
connectivity algorithms; there were no prior experimental studies on vertex
connectivity algorithms\footnote{The experimental work by
  \cite{Rigat12} mentioned $k$-vertex connectivity problem. However, in the
  experiment, they studied only the algorithm for deciding $(k,s,t)$-vertex connectivity where the source $s$ and sink $t$ are given as inputs.}.   This is in stark contrast to the
edge-connectivity problem (which is considered as a sibling problem)
where we compute the minimum number of edges to be removed to
disconnect the graph. For edge-connectivity, there are many
experimental studies
\cite{JungerR00,EdgeCutAlgos,PadbergR90,HenzingerNSS18}. More
recently, the work by \cite{ExerimentDirectedEdgeConn} implemented  the local search
framework in \cite{Forster} to compute directed edge-connectivity.   


\section{Preliminaries} \label{sec:prelim}

Let $G = (V,E)$ be an undirected graph.  In general, we denote $m = |E|$ and $n = |V|$. We denote $E(S,T)$ be the set
of edges from vertex-set $S$ to vertex-set $T$.  We say that $S \subset V$ is a
\textit{vertex cut} if $G - S$ (the graph after removing $S$ from $G$) is disconnected.
If no vertex cut of size $k$ exists, the graph is k-(vertex)-connected.
We say that $S$ is an
$xy$-vertex cut if $x$ cannot reach $y$ in $G - S$. Let $\kappa_G$ be
vertex connectivity of $G$, i.e., the size of the minimum vertex-cut
(or $n-1$ if no cut exists). Let $\kappa_G(x,y)$ denote the size of the minimum
$xy$-vertex cut in $G$ or $n-1$ if the $xy$-vertex cut does not
exist. We say that a triplet $(L,S,R)$ is a \textit{separation triple}
if $L, S$ and $R$ form a partition of $V$, $L$ and $R$ are not
$\emptyset$ and $E(L,R) = \emptyset$.  In this case, $S$ is a vertex-cut
in $G$. The decision problem for vertex connectivity which we call
$k$-connectivity problem is the following:
Given $G = (V,E)$, and integer $k$, decide if $G$ is $k$-connected,
and if not, output a vertex-cut of size $< k$.  

\textbf{Sparsification.} For an undirected graph $G = (V,E)$, the algorithm by
Nagamochi and Ibaraki \cite{sparsification} runs in $O(m)$ time and partitions $E$ into a sequence of forests $E_1,
\ldots, E_n$ (possibly $E_i = E_{i+1} = \ldots = E_n = \emptyset$ for
some $i$).  For each $k \leq n$,  the subgraph $FG_k := (V,
\bigcup_{i \leq k} E_i)$ has the property that $FG_k$ is $k$-connected
if and only if $G$ is $k$-connected. 
Moreover, any vertex cut of size $< k$ in $FG_k$ is also a vertex cut in $G$. Clearly, $|E(FG_k)| \leq nk$. 

From now, with preprocessing in $O(m)$ time, we assume that the input graph to the $k$-connectivity
problem is $FG_{k}$. In particular, we can assume that the number of
edges is $O(nk)$. We can also assume that the minimum degree is
at least $k$ (because otherwise we can output the neighbor of the
vertex with minimum degree).

\textbf{Split Graph.} The split graph construct is a standard reduction from vertex connectivity
based problems to edge connectivity based problems, used in the algorithms featured in this
paper, among others \cite{Splitgraph, Forster, HRG}. Given graph $G$, we define the split graph $SG$ as follows.
For each vertex $v$ in $G$, we replace $v$ with an ``in-vertex'' $v_{\textin}$ and an ``out-vertex'' $v_{\textout}$,
and add an edge from $v_{\textin}$ and $v_{\textout}$. The reduction follows from the observation that
edge-disjoint paths in $SG$ that start at an outvertex and end at an invertex correspond to (non-endpoint)
vertex-disjoint paths in $G$. For each edge $(u,v)$ in $G$, we add an edge from $(v_{\textin},u_{\textout})$ in $SG$.

\section{LocalEC Algorithms and Degree Counting Heuristics} \label{sec:localec}

In this section, we review two variants of LocalEC algorithms by
\cite{Forster}, and describe their corresponding new version using the
degree counting heuristic.
For completeness, we describe the complete vertex connectivity algorithm by
\cite{Forster} and some implementation details in \Cref{appendix:fullocalvc}.  
All the algorithms in this section follow a common
framework called \textsc{AbstractLocalEC} as described in \Cref{alg:abstract localec}.
Let $G = (V,E)$ be the graph that we work on. The algorithm takes as inputs $x \in V$
and two integers $\nu, k$. The basic idea is to apply
Depth-first Search (DFS) on the starting vertex $x$ but force early termination. We repeat
for $k$ iterations. If DFS terminates normally at some iteration, i.e., without having to apply the early
termination condition, then the set of reachable vertices satisfy \Cref{eq:local cut exists}.
Otherwise, we certify that no cut satisfying \Cref{eq:local cut exists} exists. The only
main difference is at line \ref{AbstractLocalEC:line1} where we need to specify the condition for early termination
and selection of the vertex $y \in V(T)$ in such a way that the entire algorithm outputs correctly
with constant probability. If the minimum degree is less than $k$, we set $k$ to the minimum degree
and return the trivial cut if no smaller cut is found.

\begin{algorithm}
  \begin{itemize}
  \item  \textbf{repeat} for $k$ times:
    \begin{enumerate}
      \item Grow a DFS tree $T$ starting from $x$, stopping early at some point to obtain $y \in V(T)$. \label{AbstractLocalEC:line1}
      \item If the DFS terminates normally, then \textbf{return} $V(T)$.  
      \item Reverse all edges along the unique path from $x$ to $y$ in the tree $T$, unless this is the last iteration.
      \end{enumerate}
      \item \textbf{return} $\bot$.
  \end{itemize}
  \caption{\textsc{AbstractLocalEC}$_G(x,\nu,k)$}
  \label{alg:abstract localec}
\end{algorithm}

Next, we define time and space complexity (in terms of edges and
vertices required to run the algorithm) of a LocalEC algorithm. 
\begin{definition}
  Let $\mathcal{A}(x,\nu,k)$ be a LocalEC algorithm. $\mathcal{A}$ has
  $(t,s_e,s_v)$-complexity if $\mathcal{A}$ terminates in $O(t)$ time and accesses at
  most $O(s_e)$ distinct edges, and at most $O(s_v)$ distinct vertices. 
\end{definition}

\subsection{Local1 and Degree Counting Version}

\textbf{Algorithm for Local1.} Replace line \ref{AbstractLocalEC:line1} in \Cref{alg:abstract localec}
with the following process. Grow a DFS tree starting on vertex
$x$ and stop when the number of accessed edges is exactly $8\nu
k$. Let $E'$ be the set of accessed edges. We sample an edge $(u,v) \in E'$ uniformly at random.
Finally, we set $y \gets u$. If we sample $(u,v)$ to be the $\tau$-th edge visited, we can stop the
DFS early after that edge (similarly to Local1+ below). 

\begin{theorem} [Theorem A.1 in \cite{Forster}] \label{thm:local1}
\textsc{Local1}$(x,\nu,k)$ is \textsc{LocalEC} with $(\nu k^2,\nu k^2, \nu k^2)$-complexity.  
\end{theorem}

Next, we present the degree counting version of Local1, which we call Local1+.

\textbf{Algorithm for Local1+.} Replace line \ref{AbstractLocalEC:line1} in \Cref{alg:abstract
  localec} with the following process. Let $\tau$ be a random integer
in the range $[1, 8\nu k]$. If this is in the last iteration, we set
$\tau \gets 8\nu k$. Then, we grow a DFS tree $T$ starting on vertex
$x$. At any time step, let $V(T)$ be the set of vertices visited by
the DFS so far. We stop as soon as $\vol^{\textout}(V(T)) \geq
\tau$. Finally, we set $y$ to be the last vertex that the DFS visited.

\begin{theorem} \label{thm:local1+}
\textsc{Local1+}$(x,\nu,k)$ is \textsc{LocalEC} with $(\nu
k^2, \nu k^2, \nu k)$-complexity.  
\end{theorem}

\subsection{Local2 and Degree Counting Version}

We say that an edge is \textit{new} if it
has not been accessed in earlier iterations. Otherwise, it
is \textit{old}. It follows that reversed edges are old.

\textbf{Algorithm for Local2.}\footnote{The algorithm Local2 described in
  this paper is similar to Algorithm 1 in \cite{Forster}. Our
  description here is simpler, and achieves the same properties as
  Algorithm 1 in \cite{Forster}.} Replace line \ref{AbstractLocalEC:line1} in \Cref{alg:abstract localec}
with the following process. We grow a DFS tree $T$ starting at vertex $x$.
Let $E'(T)$ be the set of new edges visited. We stop as soon as $E'(T) \geq 8\nu$.
Let $(u,v)$ be a random edge in $E'(T)$. Finally, we set $y \gets u$.
We do not need to store $E'(T)$ to sample from if we sample $\tau$ in the
range $[1, 8\nu]$ and choose the $\tau$-th new edge.

\begin{theorem} [Equivalent to Theorem 3.1 in \cite{Forster}] \label{thm:local2}
\textsc{Local2}$(x,\nu,k)$ is \textsc{LocalEC} with $(\nu
k^2, \nu k, \nu k)$-complexity.
\end{theorem}

Next, we present the degree counting version of Local2, which we call
Local2+. The algorithm is slightly more complicated. We set up
notations. For each $v \in V$, let $c(v)$ be the remaining capacity
for $v$, representing uncounted edge volume. Initially, $c(v) = \textdeg^{\textout}(v)$.

\textbf{Algorithm for Local2+.} Replace line \ref{AbstractLocalEC:line1} in \Cref{alg:abstract
  localec} with the following process. Let $\tau$ be a random integer
in the range $[1, 8\nu k]$. We grow a DFS tree starting on
vertex $x$. At any time step, let $v_1, v_2, ..., v_i$ be the sequence of vertices
visited by the DFS so far.
For the first vertex where $\sum_{j \leq i} c(v_j) \geq \tau$, we set $y \gets v_i$.
As soon as $\sum_{j \leq i} c(v_j) \geq 8\nu$,
we stop the DFS and update the remaining capacity $c(v)$ on each $v$ as follows.
We set $c(v_j) \gets 0$ for all $j < i$ and set $c(v_i) \gets \sum_{j \leq i} c(v_j) - 8\nu$.

Intuitively, we collect previously uncounted outgoing edges and choose the origin vertex for one of them at random.

\begin{theorem} \label{thm:local2+}
\textsc{Local2+}$(x,\nu,k)$ is \textsc{LocalEC} with $(\nu k^2, \nu k, \nu)$-complexity.  
\end{theorem}

\subsection{Proof of \Cref{thm:local1,thm:local1+,thm:local2,thm:local2+}} \label{sec:proofsinmainpart}
In this section, we address proofs for \Cref{thm:local1,thm:local1+,thm:local2,thm:local2+}.

\textbf{Correctness.} It can be shown that all four algorithms (Local1, Local1+, Local2, Local2+)
are LocalEC through a similar argument as used in \cite{Forster}. For completeness, we
provide the proofs in \Cref{appendix:omitted proofs}.

\textbf{Complexity.} Let $\mathcal{A}$ be an LocalEC algorithm (\Cref{def:localec}), and
let $\nu$, and $k$ be the parameters of the algorithm. We define three measure of
complexity $T(\cA,G),U_E(\cA,G),$ and $U_V(\cA,G)$ on input graph $G$ and LocalEC algorithm $\cA$ as follows. Let
$T(\cA,G)$ be the number of times that the algorithm accesses edges on the input graph
$G$. $T(\cA,G)$ measures time complexity of the algorithm. Let $U_E(\cA,G)$ be the
number of unique edges accessed by the algorithm on graph $G$. This measures how
much information (in terms of number of edges) that  the algorithm needs to run. Let
$U_V(\cA,G)$ be the number of unique vertices accessed by the algorithm on graph $G$. 

\begin{observation}
For any graph $G$ and LocalEC algorithm $\cA$, $T(\cA,G) \geq U_E(\cA,G) \geq U_V(\cA,G)$.
\end{observation}

\textbf{Local1.} To see that Local1 has $(O(\nu k^2), O(\nu k^2),
O(\nu k^2))$-complexity, it is enough to prove that $T(\text{Local1},G) = O(\nu k^2)$.
This follows easily because each iteration we stop the DFS after visiting
exactly $8\nu k$ edges, and there are at most $k$ iterations. 

\textbf{Local1+.} We first prove that $T(\text{Local1+},G) = O(\nu k^2)$. Since there
are $k$ iterations, it is enough to bound one iteration. Let $S$ be the set of vertices visited
by the DFS before the step at which it stops early. Clearly, $\vol^{\textout}(S) < 8\nu k$,
or we would have stopped earlier. By design, new edges can be only visited within the set
$E(S,S)$ or at the last step. Therefore, the number of edges visited is at most $|E(S,S)|+1
\leq \vol^{\textout}(S)+1 = O(\nu k)$ per iteration and $O(\nu k^2)$ in total. We have
$T(\text{Local1+},G) = O(\nu k^2)$.

Remember that if the minimum degree is initially at least $k$ to avoid trivial cuts. When
paths are reversed, no vertex other than $x$ will have reduced degree. Therefore we have
$k(|S|-1) \leq \vol^{\textout}(S) < 8\nu k$. It follows that we visit at most $O(\nu)$ vertices
in each iteration and $O(\nu k)$ in total.

\textbf{Local2.} We first prove that $U_E(\text{Local2},G) = O(\nu k)$. By design, for each
iteration, we collect at most $8\nu$ new edges. Since we repeat for $k$ iterations, we collect at
most $8\nu k$ total new edges. Next, we prove $T(\text{Local2},G) = O(\nu k^2)$. Since each edge
can be revisited at most $k$ times, we have $T(\text{Local2},G) \leq k U_E(\text{Local2},G) =
O(\nu k^2)$.

\textbf{Local2+.} We first prove that $U_E(\text{Local2+},G) = O(\nu k)$. If true, then we also
have $T_E(\text{Local2+},G) \leq k U_E(\text{Local2+},G) = O(\nu k^2)$. We will never visit an
outgoing edge of vertex $v$ unless all its capacity has been exhausted. Therefore the total used
capacity (at most $k$ times $8\nu$) is an upper bound for the number of distinct edges visited.
For $U_V(\text{Local2+},G)$, fix any iteration.
Let $S$ be the set of vertices visited by the DFS one step before
terminating and $S' \subseteq S$ the subset
of $S$ that have not been visited before. Clearly, we have
$k|S'| \leq \vol^{\textout}(S') = \sum_{v \in S'} c(v) \leq \sum_{v
  \in S} c(v) < 8\nu$. The first inequality follows since the minimum
degree is at least $k$.
We visit at most $|S'|+1 = O(\nu / k)$ distinct vertices per iteration for a total of $O(\nu)$ distinct vertices.

\section{Experimental Results} \label{sec:experiment}

\subsection{Experimental Setup}

The algorithms were implemented and compiled using C++17 with Microsoft Visual Studio 2019. All
experiments were run on a Windows 10 computer with Intel i7-9750H CPU (2.60GHz) and 16 GB DDR4-2667
RAM.

Four algorithms are compared. \textbf{LOCAL1}, \textbf{LOCAL1+} and \textbf{LOCAL2+} are
implementations based on the algorithm by Forster et al \cite{Forster}. The full algorithm to
compute vertex connectivity using LocalEC is described in \Cref{appendix:fullocalvc} and
originally by \cite{Nanongkai}. LOCAL1 and LOCAL1+ use Local1 and Local1+ as their LocalEC algorithm
with $2 \nu k$ substituted for $8 \nu k$. LOCAL2+ uses the LocalEC algorithm Local2+ with $3 \nu$
substituted for $8 \nu$. \textbf{HRG} is an implementation of the randomised version of the
algorithm by Henzinger, Rao and Gabow \cite{HRG}. The implementation
details are described in \Cref{appendix:hrg}. All algorithms were implemented using parameters
that bound theoretical success probability from below by a roughly
equal constant.  Since the data consists
of undirected graphs only, the sparsification algorithm by Nagamochi and Ibaraki
\cite{sparsification} is used together with each algorithm. The $O(m)$ partitioning of the edges
into disjoint forests is not included in the measured time. Construction of the sparse graphs in
$O(nk)$ time is included. As a result, none of the algorithms have time complexity dependent on m.
Graph size is reported only in terms of vertices.

\subsubsection{Data}

The data consists of random graphs with planted vertex cuts, random hyperbolic graphs and real world
data.

The first artificial dataset consists of graphs with a planted unique minimum vertex cut, which can
be generated with full control over vertex connectivity and balancedness. We
partition a complete graph into three sets $L$, $S$ and $R$ and use a subset of the edges in $E
\setminus E(L, R)$, chosen using a modified version of the sparsification algorithm by Nagamochi and
Ibaraki \cite{sparsification}. Like Nagamochi and Ibaraki, we label the edges to partition them into
disjoint forests $\{E_1, E_2, ... \}$ such that $(x, y) \in E_i$ implies that there is a path
between $x$ and $y$ in $E_1, E_2, ..., E_{i-1}$. Nagamochi and Ibaraki show that if this property
holds for all edges, then the union of the $k$ first forests is $k$-connected if the original graph
is $k$-connected. Unlike Nagamochi and Ibaraki, we randomly partition the edges by placing them in the
applicable forest with the lowest index in a random order.
We choose $k = 60 > |S|$ to guarantee that $S$ is a unique vertex cut that separates $L$ from $R$.
For each set of parameters we generate five graphs and run the algorithm five times each and report
the average.

The second artificial dataset consists of random hyperbolic graphs, generated using NetworKIT
\cite{NetworKIT}, which provides an implementation of the generator by von Looz et al.
\cite{vonLooz}. The properties of random hyperbolic graphs include a degree distribution that
follows a power law and small diameter, which are common in real world graphs \cite{AboutRHG}. The
graphs are generated with average degree 32 and a power law exponent of 10. We generate 20 graphs
each for sizes $2^{10}, 2^{11}, ..., 2^{18}$ vertices and group them according to vertex
connectivity. We run the algorithm five times per graph and report the average for each group with
the same size and vertex connectivity.

The real world data is based on three graphs from the SNAP dataset \cite{SNAP}, soc-Epinions1,
com-LiveJournal and web-BerkStan. The LiveJournal dataset is originally undirected. The other two
are directed graphs read as undirected, which means that we compute weak vertex connecitivity for
these graphs. We preprocess these graphs by taking the largest connected component for a $k$-core. A
$k$-core is defined as the edge-maximal subgraph with minimum degree at least $k$. Only $k$-cores whose
vertex connectivity is over 1 but less than the minimum degree are used. For each $k$-core we run the
algorithms 25 times and report the average.

\subsection{Planted Cuts} \label{sec:planted cuts}

In theory the running time for HRG is linear in $\kappa$ and the algorithms based on Forster et al.
\cite{Forster} are cubic in $\kappa$. Figure \ref{PC_VarLS:a} shows that the running time for HRG indeed
grows much slower with $\kappa$. The running time for LOCAL1 exceeds that of HRG much earlier, at $\kappa \geq
17$, than LOCAL1+ ($\kappa \geq 40$) and LOCAL2+ ($\kappa \geq 48$).

Figure \ref{PC_VarLS:b} shows that all four algorithms perform reasonably well both for graphs with
unbalanced cuts and balanced cuts, although HRG is faster for unbalanced graphs by a factor of 2.
Internal testing suggests that the running time of HRG is roughly proportional to $|L|^2 + |R|^2$.
The difference between the highest and lowest running time is
a factor of 1.99 for HRG, 1.19 for LOCAL1, 1.27 for LOCAL1+ and 1.16 for LOCAL2+.

\begin{figure}[ht]
\centering
	\captionsetup{position=top,justification=centering}
	\subcaptionbox{$|L| = 5$, $n = 1000$\label{PC_VarLS:a}}
		{\includegraphics[width=0.3\linewidth]{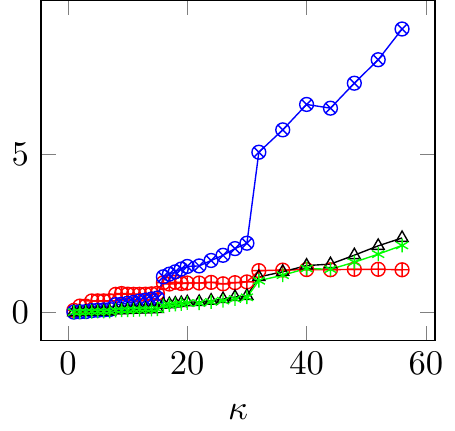}}
	\subcaptionbox{$\kappa = 5$, $n = 1000$\label{PC_VarLS:b}}
		{\includegraphics[width=0.3\linewidth]{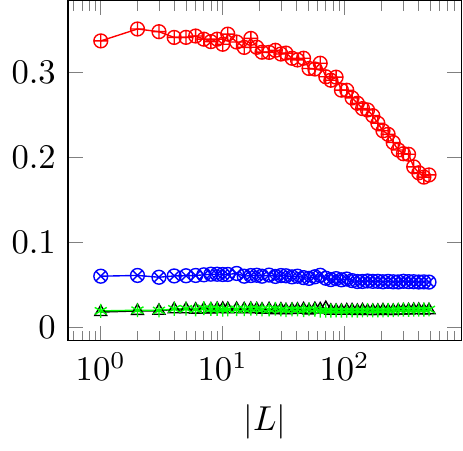}}
	\subcaptionbox*{}{\includegraphics[width=0.2\linewidth]{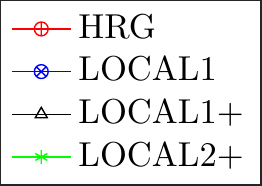}}
\caption{Running time (seconds) for Planted Cuts with variable $|L|$ or $\kappa$}
\label{PC_VarLS}
\end{figure}

When $\kappa < 16$, LOCAL1, LOCAL1+ and LOCAL2+ outperform the quadratic-time HRG on very small
graphs with planted cuts. At $\kappa = 4$ in figure \ref{PC_varNperN:a}, HRG takes 23 ms for 100
vertices, which is already slower than both LOCAL1+ and LOCAL2+. LOCAL1 is faster than HRG at $n
\geq 200$. When $\kappa = 15$ (figure \ref{PC_varNperN:d}), HRG is slower than LOCAL1+ and LOCAL2+
at $n \geq 250$ and LOCAL1 at $n \geq 550$.

LOCAL1+ and LOCAL2+ perform very similarly for small graphs but on larger graphs, LOCAL2+ is faster,
as shown by figure \ref{PC_varNperN:e}.

\begin{figure}[ht]
\centering
	\captionsetup{position=top,justification=centering}
	\subcaptionbox{$\kappa = 4$, $|L| = 5$\label{PC_varNperN:a}}
		{\includegraphics[width=0.3\linewidth]{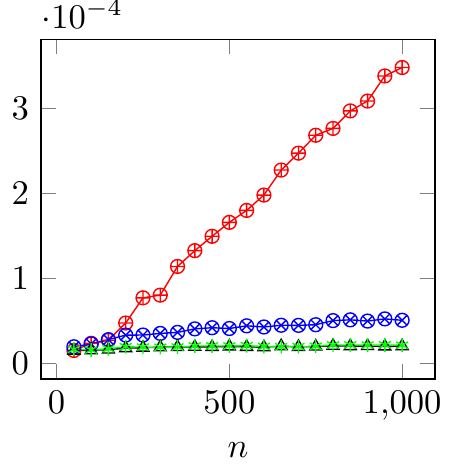}}
	\subcaptionbox{$\kappa = 7$, $|L| = 5$}
		{\includegraphics[width=0.3\linewidth]{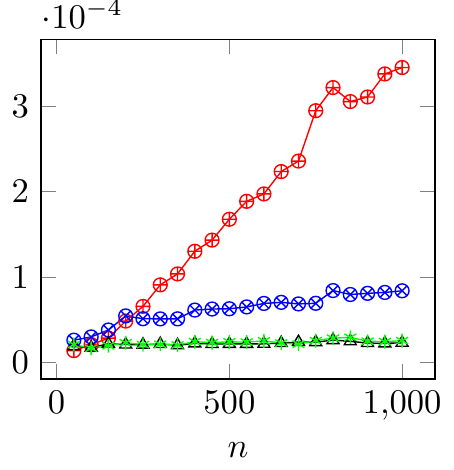}}
	\subcaptionbox{$\kappa = 8$, $|L| = 5$\label{PC_varNperN:c}}
		{\includegraphics[width=0.3\linewidth]{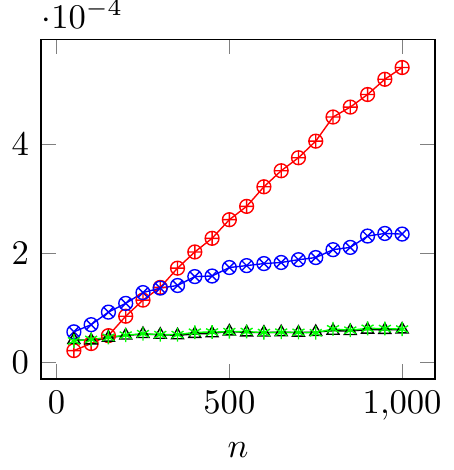}}
	\subcaptionbox{$\kappa = 15$, $|L| = 5$\label{PC_varNperN:d}}
		{\includegraphics[width=0.3\linewidth]{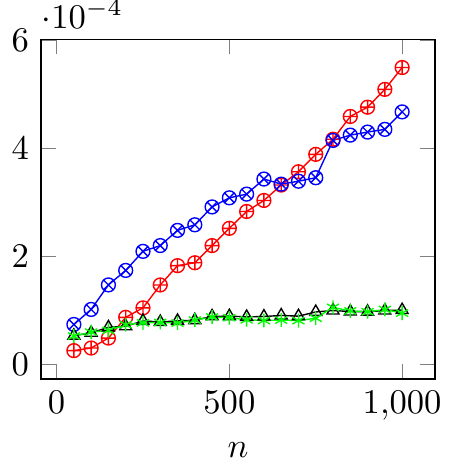}}
	\subcaptionbox{$\kappa = 31$, $|L| = 5$\label{PC_varNperN:e}}
		{\includegraphics[width=0.3\linewidth]{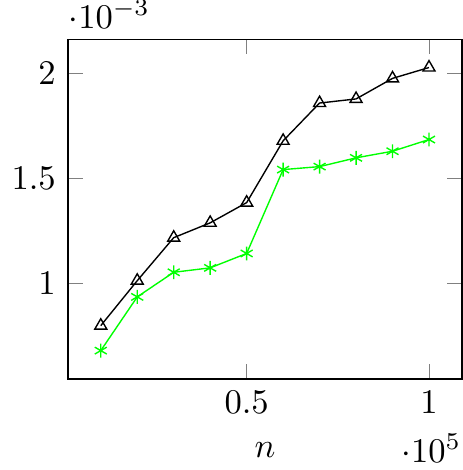}}
	\subcaptionbox*{}{\includegraphics[width=0.2\linewidth]{legend.pdf}}
\caption{Running time (seconds) per vertex for Planted cuts with fixed $|L|$ and $|S|$}
\label{PC_varNperN}
\end{figure}

\subsection{Random Hyperbolic Graphs} \label{sec:random hyperbolic graphs}

HRG is much faster on random hyperbolic graphs than on the planted cut dataset. Comparing figures
\ref{PC_varNperN:c} and \ref{RHG_varNperN:c}, the performance of HRG on 1000 vertex graphs with
planted cuts of size 8 is similar to that on random hyperbolic graphs with the same vertex
connectivity and over 30000 vertices. The performance differences are smaller for LOCAL1, LOCAL1+
and LOCAL2+, which means that the point at which these algorithms outperform HRG occurs at somewhat
higher $n$.

For random hyperbolic graphs with $\kappa = 7$ (figure \ref{RHG_varNperN:b}), HRG and LOCAL1 are
equally fast at 4096 vertices (0.6 seconds). HRG is faster than LOCAL1 for all included random
hyperbolic graphs where $\kappa > 7$, including graphs up to 32768 vertices. The running time for
LOCAL1+ and LOCAL2+ is close to that of HRG for random hyperbolic graphs where $\kappa = 12$ and $n
\in [1024, 4196]$ (figure \ref{RHG_varNperN:e}).

\begin{figure}[ht]
\centering
	\captionsetup{position=top,justification=centering}
	\subcaptionbox{$\kappa = 4$}
		{\includegraphics[width=0.3\linewidth]{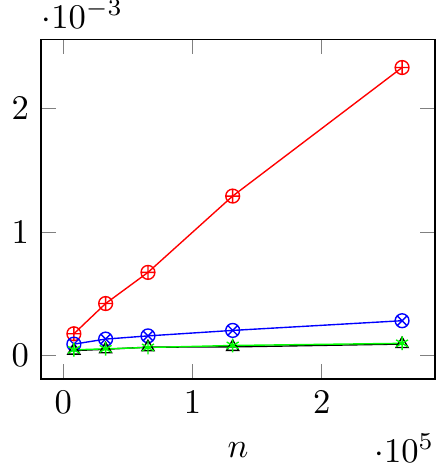}}
	\subcaptionbox{$\kappa = 7$\label{RHG_varNperN:b}}
		{\includegraphics[width=0.3\linewidth]{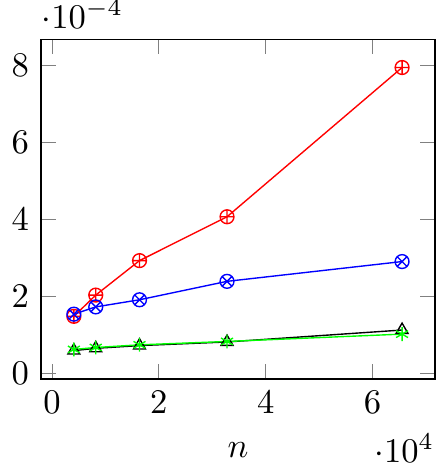}}
	\subcaptionbox{$\kappa = 8$\label{RHG_varNperN:c}}
		{\includegraphics[width=0.3\linewidth]{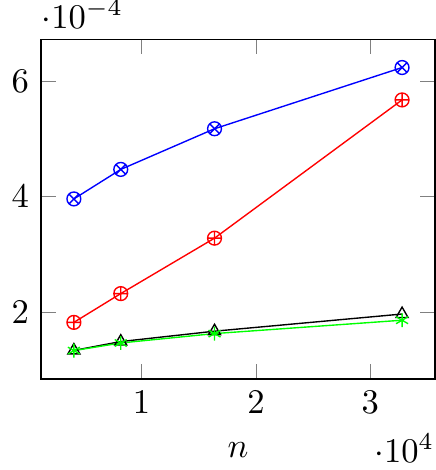}}
	\subcaptionbox{$\kappa = 10$}
		{\includegraphics[width=0.3\linewidth]{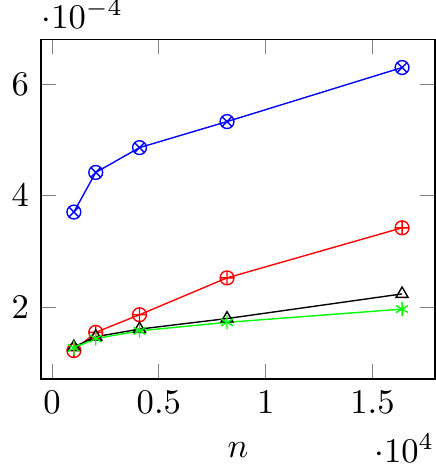}}
	\subcaptionbox{$\kappa = 12$\label{RHG_varNperN:e}}
		{\includegraphics[width=0.3\linewidth]{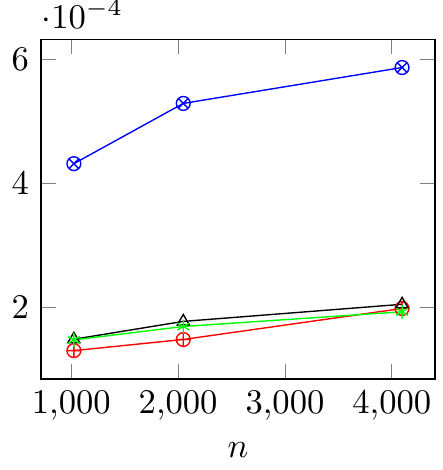}}
	\subcaptionbox*{}{\includegraphics[width=0.2\linewidth]{legend.pdf}}
\caption{Running time (seconds) per vertex for Random Hyperbolic Graphs}
\label{RHG_varNperN}
\end{figure}

\subsection{Real-World Networks} \label{sec:real world}

Figure \ref{table-RWG} presents real world network data. Each row represents a $\delta$-core, where
$\delta$ is the minimum degree of the resulting graph. Note that in general, minimum degree for a
$k$-core can exceed $k$. Figure \ref{table-PC} shows data for graphs with planted cuts with similar
parameters to the real world graphs, for comparison.

LOCAL1+ and LOCAL2+ clearly outperform LOCAL1 on real-world networks, as on artificial data.
The $k$-cores of soc-Epinions1 have very similar performance in
real world networks and graphs with planted cuts in figure \ref{table-PC}.
Performance for other real network data is generally faster for all four algorithms than for planted
cuts, especially for HRG, which is 5-8 times faster on real world data. Similarly, running times for
LOCAL1, LOCAL1+ and LOCAL2+ are also higher on random hyperbolic graphs than on $k$-cores
of com-lj.ungraph and web-BerkStan.

\begin{figure}[ht] 
	\includegraphics{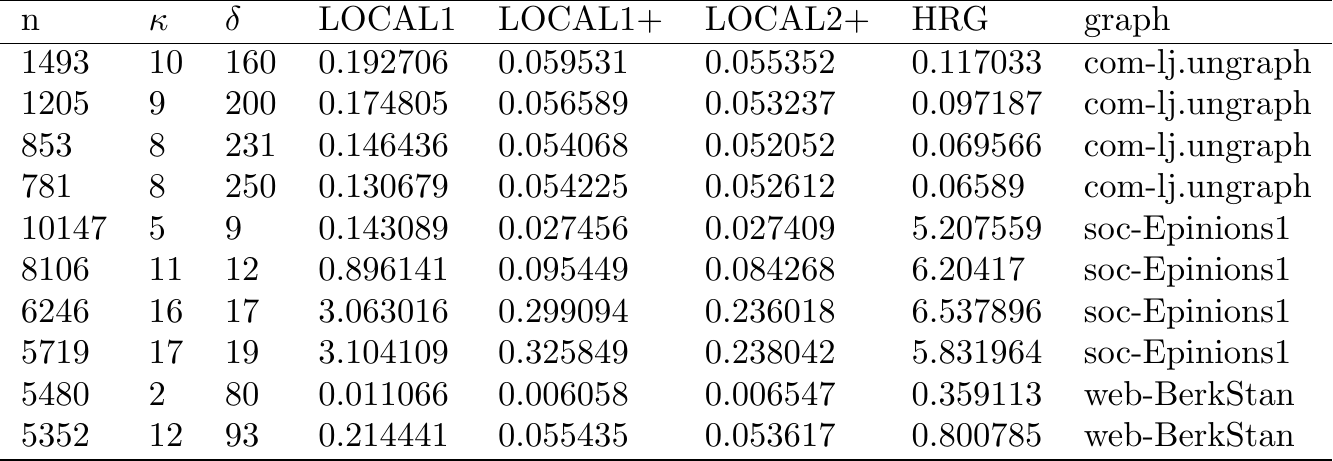}
\caption{Running times (milliseconds) per vertex on $k$-cores for real world networks}
\label{table-RWG}
\end{figure}

\begin{figure}[ht] 
	\includegraphics{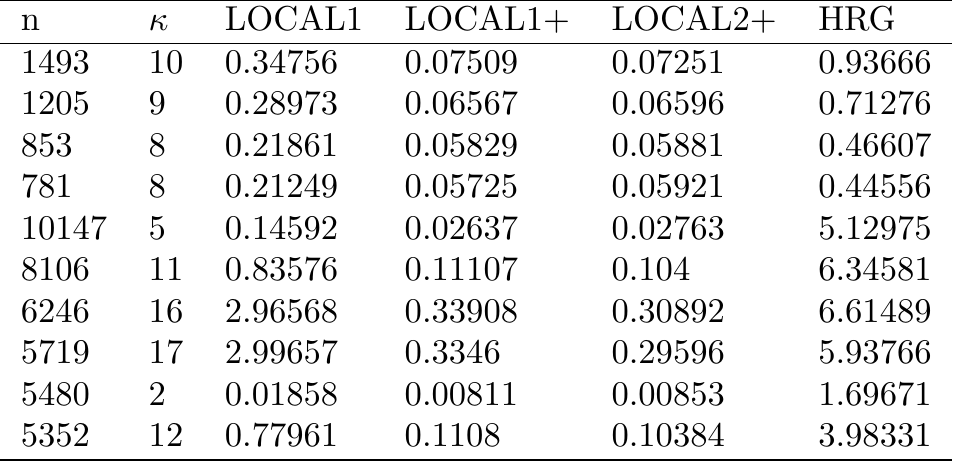}
\caption{Running times (milliseconds) per vertex on Planted Cuts ($|L| = 5$)}
\label{table-PC}
\end{figure}

\begin{figure}[h!]
\raggedright
	\captionsetup{position=top,justification=centering}
\subcaptionbox{k=8, $\frac{E_{LocalEC}}{\nu k}$}
{\includegraphics[width=0.3\linewidth]{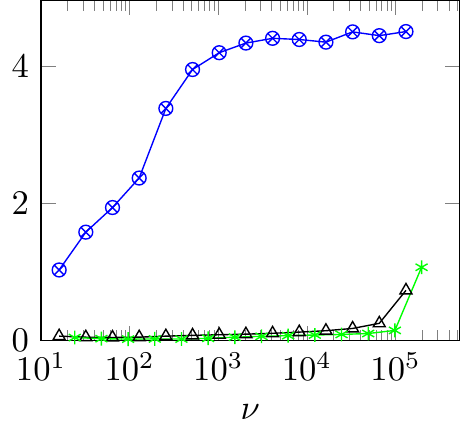}}
\subcaptionbox{k=16, $\frac{E_{LocalEC}}{\nu k}$}
{\includegraphics[width=0.3\linewidth]{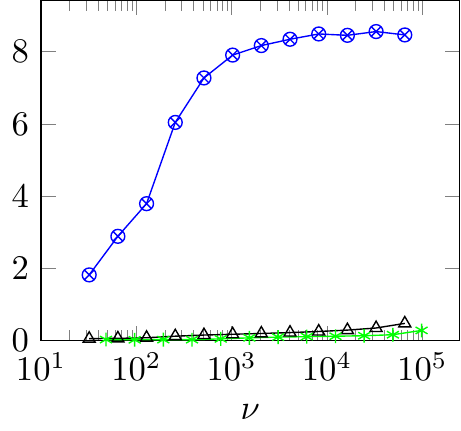}}
\subcaptionbox{k=32, $\frac{E_{LocalEC}}{\nu k}$}
{\includegraphics[width=0.3\linewidth]{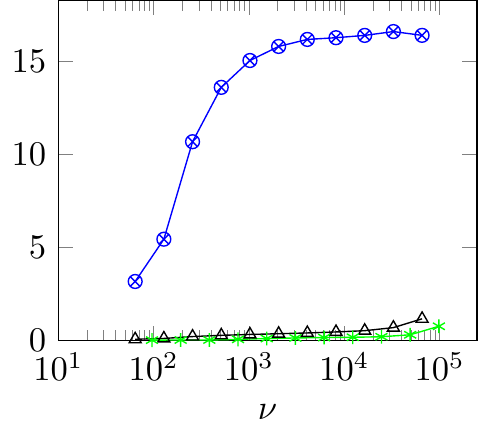}}
\subcaptionbox{k=8, $\frac{E_{LocalEC}}{\nu k}$}
{\includegraphics[width=0.3\linewidth]{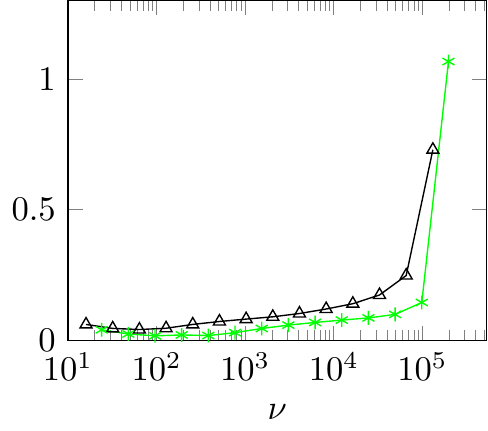}}
\subcaptionbox{k=16, $\frac{E_{LocalEC}}{\nu k}$}
{\includegraphics[width=0.3\linewidth]{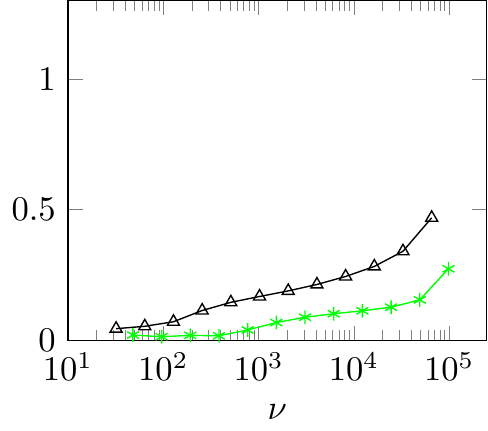}}
\subcaptionbox{k=32, $\frac{E_{LocalEC}}{\nu k}$}
{\includegraphics[width=0.3\linewidth]{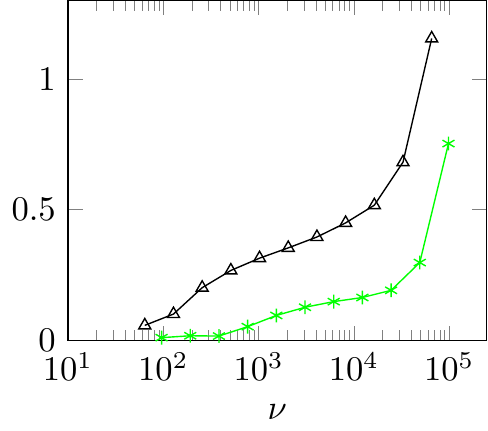}}
\subcaptionbox*{}{\includegraphics[width=0.2\linewidth]{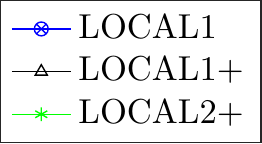}}
\caption{Planted cuts with $n = 100000, |L|=5, k = \kappa$\\
(Non-unique) average edges per LocalEC call, normalised by $\nu k$}
\label{edgesovervk}
\end{figure}

\subsection{Effectiveness of Degree Counting} \label{sec:effectiveness}

In figure \ref{edgesovervk} we study internal measurements from LocalEC in the different algorithms.
Note that Local1 and Local1+ apply a multiplicative factor of 2 to $\nu$ and Local2+ a factor of 3.
The values used here include
this increase. The number of edges visited by the average call to LocalEC at each value for $\nu$ is
normalised by $\nu k$. This metric approximately doubles for Local1 and Local1+ when $k$ is doubled, as expected for
algorithms quadratic in $k$. The metric grows for Local2+ too, but by a smaller factor around 1.5
for most values. The growth is faster for higher values for $\nu$ and for the highest values it is
approximately by a factor 2, like the other two algorithms.

The number of edges explored relative to $\nu$ is higher for high $\nu$ for all algorithms and
parameters in figure \ref{edgesovervk}. For LOCAL1, it converges towards $\frac{\nu k}{2}$, which is
the average of $[1, \nu k]$, the range of possible early stopping points $\tau$.

 Local1 clearly visits more edges than in Local1+ and Local2+ by a large factor, according to figure
\ref{edgesovervk}. Figure \ref{table-CPU} shows that most of the running time of LOCAL1 is used
searching for unbalanced cuts with LocalEC. However, LOCAL1+ and LOCAL2+ spend a similar amount of
time on balanced and unbalanced cuts. The only difference between the versions is the choice of
LocalEC. These results suggest that degree counting improves the practical performance of LocalEC
significantly but there is not much more room for improvement through LocalEC without also further
optimising x-y max flow to search for balanced cuts. When the number of vertices is increased by a
factor of 10, the time spent searching for unbalanced cuts does not seem to grow faster than the
time spent searching for balanced cuts. The category ``other'' is dominated by initial setup for the
data structures.

\begin{figure}[h!]
	\captionsetup{position=top,justification=centering}
\subcaptionbox{$n=10000,
\kappa=8$\label{table-CPU:subA}}{\includegraphics[width=0.49\textwidth]{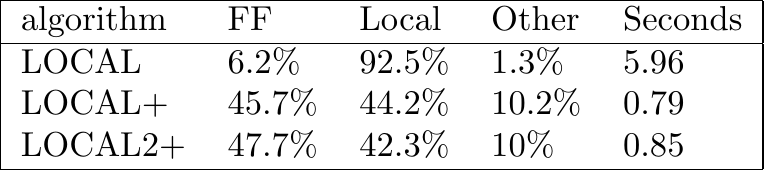}}
\subcaptionbox{$n=10000,
\kappa=16$\label{table-CPU:subB}}{\includegraphics[width=0.49\textwidth]{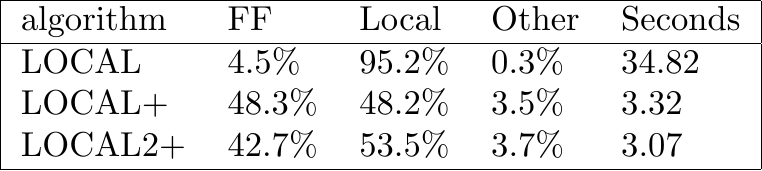}}\\
\subcaptionbox{$n=100000,
\kappa=8$\label{table-CPU:subC}}{\includegraphics[width=0.49\textwidth]{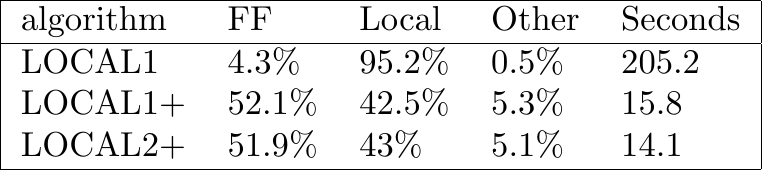}}
\subcaptionbox{$n=100000,
\kappa=16$\label{table-CPU:subD}}{\includegraphics[width=0.49\textwidth]{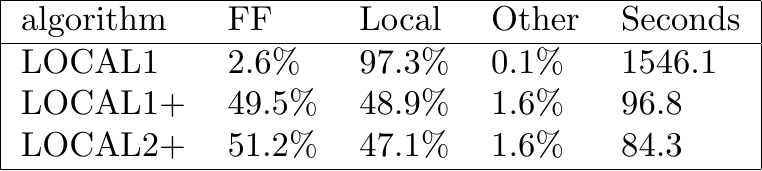}}
\caption{CPU use: balanced cuts/Ford-Fulkerson(FF) vs unbalanced cuts/LocalEC(Local).
\\Running time was measured separately.}
\label{table-CPU}
\end{figure}

\subsection{Success rate}

We define the success rate of a vertex connectivity algorithm as the percentage of attempts that
yields an optimal cut. The observed success rate for HRG is at or near 100\% on all featured
datasets. For random hyperbolic graphs, none of the algorithms returned nonoptimal cuts. For graphs
with planted cuts and $k$-cores of real world networks, the success rates are 97\%+ for LOCAL1, 96\%+
for LOCAL1+ and 95\%+ for LOCAL2+.

\section{Conclusion and Future Work}  \label{sec:conclusion}

We study the experimental performance of the near-linear time algorithm by
\cite{Forster} when the input graph connectivity is small. The algorithm is
based on local search. We also introduce a new heuristic for the local search
algorithm, which we call degree counting. Based on experimental results, the
degree counting heuristic significantly improves the empirical running time of
the algorithm over its non-degree counting counterpart. For future work, we
plan to extend the experiments to directed graphs, and on larger instances
of datasets (order of millions edges). 

\bibliography{sources}


\appendix

\section{Omitted Proofs}  \label{appendix:omitted proofs}
\subsection{Correctness}

To show that any algorithm among \textsc{Local1}, \textsc{Local1+}, \textsc{Local2}, and
\textsc{Local2+} is LocalEC, it is enough to prove that it satisfies two properties:

\begin{property} \label{prop:1}
If $V(T)$ is returned, then  $|E(V(T), V - V(T))| < k$ and $\emptyset \neq V(T) \subsetneq V$.
\end{property}

\begin{property} \label{prop:2}
If there is a vertex-set $S$ satisfying \Cref{eq:local cut exists}, then $\bot$ is returned with
probability at most $1/2$.
\end{property} 

The following simple observation is due to \cite{Chechik}. 
\begin{observation} \label{obs:dfs}
Let $S$ be a vertex-set in graph $G$ and $x \in S$. Let $P$ be a path from $x$ to $y$. Let $G'$ be
$G$ after reversing all edges along $P$. If $y \in S$, then $|E_{G'}(S,V-S)| = |E_{G}(S,V-S)|$.
Otherwise, $|E_{G'}(S,V-S)| = |E_{G}(S,V-S)| - 1$.
\end{observation}

For the first property, the following argument works for all four algorithms. 
\begin{lemma}
 \textsc{Local1}, \textsc{Local1+}, \textsc{Local2} and \textsc{Local2+} satisfy \Cref{prop:1}.
\end{lemma} 
\begin{proof}
Let $S = V(T)$ be the cut the algorithm returned. Observe that $x \in S$ by design.
By \Cref{obs:dfs}, each iteration can only reduce the number of crossing edges by at most
one. This can happen at most $k-1$ times before the final iteration, which implies that initially'
$|E(S,V-S)| \leq k-1$.
\end{proof}

For the second property, the following argument works for \textsc{Local1}, and \textsc{Local1+}
\begin{lemma}
 \textsc{Local1} and \textsc{Local1+} satisfy \Cref{prop:2}. 
\end{lemma}
\begin{proof} 
We focus on proving that LOCAL1 satisfies \Cref{prop:2} (the proof for Local1+ will be essentially
identical). If the algorithm terminates before the $k$-th iteration, then it outputs $V(T)$, and
thus $\bot$ is never returned. So now we assume that the algorithm terminates at the $k$-th
iteration. Let $y_1, \ldots y_{k-1}$ be the sequence of chosen path endpoints $y$ in DFS iterations.
We first bound the probability that $y_i \in S$. Let $\vol^{\textout}_i(S)$ be the volume of $S$ at
iteration $i$. So,
\begin{equation} \label{eq:prob y in S}
\qquad \qquad \qquad \quad \Pr(y_i \in S) \leq \frac{\vol^{\textout}_i(S)}{8\nu k} \leq
\frac{\vol^{\textout}(S)}{8\nu k} \leq \frac{\nu}{8\nu k} = \frac{1}{8k}.
\end{equation} 

The first inequality follows by design. The second inequality follows by \Cref{obs:dfs}. 

By \Cref{obs:dfs}, the algorithm can only return $\bot$ at the final iteration if at least one
of the $y_i$'s is in $S$ (or if there is not viable cut). Let $\indicator[y_i \in S]$ be an indicator
function. Let $Y = \sum_{i
\leq k-1} \indicator[y_i \in S]$. Observe that $Y \geq 1$ if and only if the algorithm outputs
$\bot$. We now bound the probability that $Y \geq 1$. By linearity of expectation, we have $\Ex[Y] =
\sum_{i \leq k-1}\Ex[\indicator[y_i \in S]] = \sum_{i \leq k-1} \Pr(y_i \in S) \leq \frac{1}{8}.$
Therefore, by Markov's inequality, we have
\begin{equation}\label{eq:prob Y geq 1}
$$ \Pr(Y \geq 1) = \Pr(Y \geq 8\cdot\frac{1}{8}) \leq \Pr(Y \geq 8\Ex[Y]) \leq \frac{1}{8}.$$
\end{equation}
This completes the proof for \textsc{LOCAL1}. To see that the same proof works for \textsc{LOCAL1+},
observe that the proof above (\Cref{eq:prob y in S} in particular) does not use the identity of the
edges. Outgoing edges of a vertex are interchangible. The degree counting variant counts edges
ensures that each outgoing edge for visited vertices is included in the collection of edges without
collecting explicitly. The precomputed random number $\tau$ corresponds to a random edge from
the collection.
\end{proof}

It remains to prove the second property for \textsc{Local2} and \textsc{Local2+}. However, the
arguments for \textsc{Local2} and \textsc{Local2+} are very similar to Local1 and Local1+:

\begin{lemma}
\textsc{Local2} and  \textsc{Local2+} satisfy \Cref{prop:2}. 
\end{lemma}
\begin{proof}
For \textsc{Local2}, each edge in $E(S, V)$ has a $\frac{1}{8\nu}$ probability to be chosen if the edge
is visited. The probabilities are not independent but can be used for Markov's inequality. If $Y$ is the
number of edges in $E(S, V)$ that are chosen, or equivalently the number of times a vertex in
$S$ is chosen, we have $\Ex[Y] \leq \frac{\nu}{8\nu} = \frac{1}{8}$, resulting in the same equation
as \Cref{eq:prob Y geq 1}. If we consider the case where all edges in $E(S, V)$ are visited in a single
iteration, we can see that the bound is tight. For \textsc{Local2+}, apply the same logic to $c(v)$ instead of edges.
\end{proof}

\section{Full Near-Linear Vertex Connectivity Algorithm} \label{appendix:fullocalvc}

\subsection{Vertex Connectivity via Local Edge Connectivity in Undirected Graphs}

In this section, we describe the vertex connectivity algorithm that we implement in this paper.
We will assume that we have a LocalEC algorithm (\Cref{def:localec}) with time complexity $O(\nu k^2)$.

Let $G$ be a directed graph with n vertices and m edges, such that $(x, y) \in E(G) \iff (y, x) \in
E(G)$. This is a directed representation of an undirected graph. Given a positive integer $k$, the
following algorithm, which is very closely based on the framework by Nanongkai et al.
\cite{Nanongkai}, finds a minimum vertex cut of size less than $k$ or certifies that $\kappa \geq k$
with constant probability. Let $k'$ be the
size of the minimum cut found so far in the algorithm, or $k$ if no cut has been found yet.

Suppose that there is a vertex cut in $G$, represented by a separation triple $(L,S,R)$.
Assume without loss of generality that
$\text{vol}^\text{out}(L)\leq\text{vol}^\text{out}(R)$. If $\text{vol}^\text{out}(L) < 2\delta$,
where $\delta$ is the minimum degree, then $|L| = 1$. We find $\delta$ and such trivial cuts with a
linear sweep.

Fix some value $a = \Theta(m/k)$, which must be a valid value for the parameter $\nu$ in LocalEC.

\textbf{Balanced Cut.} Suppose that $\text{vol}^\text{out}(L) \geq a$. If we sample pairs of edges
$(x, x'), (y, y') \in E(G)$ we can show that $x \in L, y \in R$ with probability $\Theta(a/m)$ for
each sample. We can find a x-y vertex cut of size less than $k'$ if one exists by using a max flow
algorithm on the split graph through a well-known reduction (e.g. \cite{Splitgraph}). A sample size
of $\Theta(m/a)$ is sufficient to find such a cut with high probability.

\textbf{Unbalanced Cut.} Now, for $\nu \in \{2^i\delta | i \in \mathbb{Z}^{\geq0}, 2^i\delta < a\}$,
i.e., power of two multiples of $\delta$ up to $a$.
we sample $\Theta(m/\nu)$ edges $(x, x') \in E(G)$ and run $\text{LocalEC}(x_\text{out}, \nu, k')$ on the
split graph for each $x$. If $\text{vol}^\text{out}(L) = \Theta(\nu)$, the probability that any
given edge yields $x \in L$ is $ \Theta(\nu/m)$, which means that a sample size of $\Theta(m/\nu)$
is sufficient to find one with high probability. Let $L' = \{x_\text{in}, x_\text{out} | x \in L\}
\cup \{x_\text{in} | x \in S\}$. $L'$ is one side of an edge cut that corresponds to the vertex cut $S$, as in
the reduction used for x-y connectivity for balanced cuts. We can show that $\text{vol}^\text{out}(L') =
\frac{k+1}{k}\text{vol}^\text{out}(L) + k = \Theta(\text{vol}^\text{out}(L))$. Clearly, if
$\text{vol}^\text{out}(L) = o(a)$, we will run LocalEC with some value $\nu$ for a sufficient sample
size to find the cut with high probability.

In practice, if $\text{vol}^\text{out}(L) = \Theta(a)$, there is a fairly high probability to find
the cut both with the max flow algorithm and LocalEC. At $\frac{a}{2}$, the max flow algorithm finds the
cut at approximately half the probability at $a$. LocalEC, when configured to find cuts with
reasonably high probability at $\nu$ will also often find cuts at higher volumes with diminishing probability
as the actual volume goes up.

If we do not start with some $k > \kappa$, we can find one by doubling $k$ until we find a cut. When
a cut can be found, a minimum cut will be find with high probability.

\textbf{Time Complexity.} Assuming Ford-Fulkerson max flow that runs in $\Theta(mk)$ time, the
running time for finding balanced cuts is $\Theta(mk)\Theta(m/(m/k)) = \Theta(mk^2)$. Assuming
LocalEC that runs in $O(\nu k^2)$ time, the running time for each of the $\Theta(\log (m/k))$ values
for the parameter $\nu$ is $\Theta(\nu k^2)\Theta(m/\nu) = \Theta(mk^2)$. Due to preprocessing by
Nagamochi and Ibaraki, which runs in $\Theta(m)$ time, we have $m = \Theta(nk)$, for a final time
complexity of $\Theta(m + k^3n\log n)$. If we repeat for high rather than constant probability we
square the logfactor.

\subsection{Implementation Details.}

We use the following numbers for the unspecified values above: $a = \frac{m}{3k}$, $\frac{m}{a} =
3k$ samples for Ford-Fulkerson and $\lfloor\frac{m}{\nu}\rfloor$ for LocalEC. For Local1 and Local1+
we collect/count $2 \nu k$ edges rather than $8 \nu k$ and for Local2+ we count to $3 \nu$ rather
than $8 \nu$. Local2+ seems to need a slightly higher factor for similar success rate.

The graph implementation used for this paper is based on adjacency lists with c++ vectors. When we
reverse edges along a path we save the relevant vector indices to enable us to perform the opposite
operations later, in order from the newest reversed path to the oldest. We store information such as
DFS visited vertex flags and the number of uncounted edges/coins in LOCAL2+ per vertex. To avoid
resetting this information for every vertex, we also maintain lists of vertices that have been
visited within the most recent DFS or LocalEC call.

\section{Preflow-push based Vertex Connectivity Algorithm} \label{appendix:hrg}

We use the algorithm by Henzinger, Rao and Gabow \cite{HRG} with only minor optimisations. We omit
most details here. The core algorithm uses a preflow based algorithm to calculate the minimum $S_i
x_i$-cut, where $S_i = \{x\} \cup \{x_j : j < i\}$, for each vertex $x_i$ not adjacent to x. The
algorithm maintains an ``awake'' set $W$ of vertices from where the current sink may be reachable.
If there exists a minimum vertex cut $S \ni x$, which is very probable for small $\kappa$, then the
minimum of these cuts will be a minimum vertex cut. The algorithm is repeated if needed to achieve
a 50\% or lower error rate, which should not be the case for any included test case.
As with the algorithms by Forster et al.
\cite{Forster}, we use the spit graph reduction and the sparsification algorithm by Nagamochi and
Ibaraki \cite{sparsification} to reduce the average degree of the graph to at most $k$, doubling $k$
until we find a cut smaller than $k$. In case of weighted edges, dynamic trees would be used to
improve time complexity, but this article only uses unweighted edges.

On page 10 of \cite{HRG}, Henzinger et al. describe a guaranteed method of doubling $k$ to find some
$k \in (\kappa, 4\kappa)$. There, the algorithm is run on an arbitrary nonrandom vertex of degree
$k$. To obtain an optimal cut with any probability guarantee, the algorithm needs to be repeated on a random
seed vertex. We use random seed vertices during doubling to avoid having to repeat the algorithm
after already finding a cut of size less than $k$. For small $k$, the ``bad case'' of not finding a
cut despite $\kappa < k$ is highly unlikely.

On page 20 of \cite{HRG}, Henzinger et al. describe multiple auxiliary data structures used to
achieve the desired time complexity. One of these is a partition of vertices in the awake set $W$ by
their current distance values. We add another auxiliary data structure that stores the index of a
vertex in this data structure to speed up finding and removing a vertex, which happened frequently
enough to create a CPU hotspot.

\end{document}